\documentclass{cccg26_withoutlogo}
\usepackage{graphicx,amssymb,amsmath}
\usepackage{cases} 
\usepackage{xcolor} 
\usepackage{ulem} 
\usepackage{algorithm}
\usepackage{algpseudocode}

\usepackage{amsmath}
\makeatletter
\newcommand{\xline}[2][]{%
  \ext@arrow 0055{\arrowfill@\relbar\relbar\relbar}{#1}{#2}%
}
\makeatother

\newtheorem{defin}[theorem]{Definition} 

\title{Finding Some Impossibility of Flat-Folding of Given Origami Crease Pattern by Graphical Representation}

\author{Chihiro Nakajima\thanks{Department of Engineering, Tohoku Bunka Gakuen University, \texttt{chihiro.nakajima@ait.tbgu.ac.jp}}} 

\index{Author, First}
\index{Researcher, Second}

\begin{document}
\thispagestyle{empty}
\maketitle

\begin{abstract}
The flat-foldability problem in origami asks whether a given crease pattern can be folded flat without any physical penetration or intrusion of polygons into the creases.  As established by Bern and Hayes, determining the global flat-foldability of a general crease pattern is an NP-hard problem. 

In this paper, we focus on unsigned crease patterns that satisfy the necessary local conditions imposed by the Kawasaki-Justin theorem at all interior vertices. To evaluate global foldability, we introduce an undirected graph representation-an overlap graph-that models pairwise non-intrusion constraints among overlapping polygons in a flattened state.  Using this graphical representation, we propose a polynomial-time algorithm to efficiently detect the impossibility of flat-folding by analyzing the algebraic properties of the graph's cycle basis. Specifically, we classify the nodes (intermediations) along each cycle and prove that the parity of a specific node kind governs the mathematical consistency of the loop. Detecting a self-inconsistent, frustrated, cycle via parity evaluation provides a robust sufficient condition for demonstrating that the entire crease pattern cannot be flat-folded. 

This result successfully isolates the tractable components of flat-foldability from its worst-case NP-hardness, providing a deeper understanding of the precise structural features that cause global computational difficulty. We also demonstrate the efficacy of our method by applying it to a well-known crease pattern that is fundamentally impossible to flat-fold.
\end{abstract}

\section{Introduction}
The flat-foldability problem asks whether a given crease pattern (e.g., Fig. 1(a)), which divides a piece of paper into polygons, can be folded flat into a valid stacking order without causing any self-intersections among the polygons and creases. Determining the flat-foldability of a general crease pattern, regardless of whether mountain/valley assignments are specified, is known to be NP-hard, as originally proven by Bern and Hayes \cite{BH} and later corrected by Akitaya et al. \cite{Akitayaetal}.

A fundamental necessary condition for flat-foldability is the Kawasaki-Justin theorem \cite{Kaw,Jus}, which restricts the angles around each interior vertex. However, while satisfying this condition guarantees local flat-foldability around every individual vertex, it does not ensure global flat-foldability, as global layer-stacking conflicts may still occur. Additional local constraints on valid mountain-valley assignments are given by the Big-Little-Big angle theorem \cite{Hull_b}.

Because of the gap between local and global foldability, proving global non-foldability often requires an exhaustive search to demonstrate that no valid stacking order exists.

In this study, we present a novel suﬃcient condition for global non-foldability applicable to unsigned crease patterns that satisfy local foldability at all interior vertices.
This condition enables a rigorous proof of global non-foldability in polynomial time, bypassing the need for exhaustive verification of stacking orders.
By identifying a computationally tractable subclass within locally flat-foldable patterns, we narrow the set of inherently hard instances in the flat-foldability decision problem and characterize the essential structural properties that govern the hardness.

To achieve this, we represent local stacking constraints and relative stacking orders graphically to efficiently detect global inconsistencies. The remainder of this paper is organized as follows.
Section 2 introduces the overlap diagram, relative stacking orders, and the overlap graph to formally encode local geometric constraints among subsets of polygons. 
Section 3 defines key concepts, such as intermediations, to analyze the structural inconsistency caused by constraint propagation along cycles within the overlap graph. Through this framework, we prove our main theorem for global non-foldability and demonstrate its application on an unsigned crease pattern which is known to be impossible to flat-fold \cite{Hull_t}. The content of this section is closely related to 2-XORSAT, and theorems in the latter half of Section 3 can be inherently linked to its properties. However, in this paper, the derivation of overlap graphs involves handling the indices of polygons forming a crease pattern in a somewhat complicated manner. To avoid any resulting confusion and to explain unambiguously that a clear origin, at least one of them, of global non-foldability can be reduced to the 2-XORSAT problem, we present these theorems using the graph-theoretic terminology introduced herein.
Next, Section 4 highlights unaddressed constraints and exceptional overlap diagrams that require an extension of our overlap graph, clarifying the fact that the absence of cycle-based inconsistency does not immediately guarantee global flat-foldability. Finally, Section 5 summarizes the findings of this paper and provides concluding remarks.

\section{Introduction of Overlap Diagram and Overlap Graph}
\begin{figure}[h!tbp]
  \centering
  \includegraphics[width=2.4in]{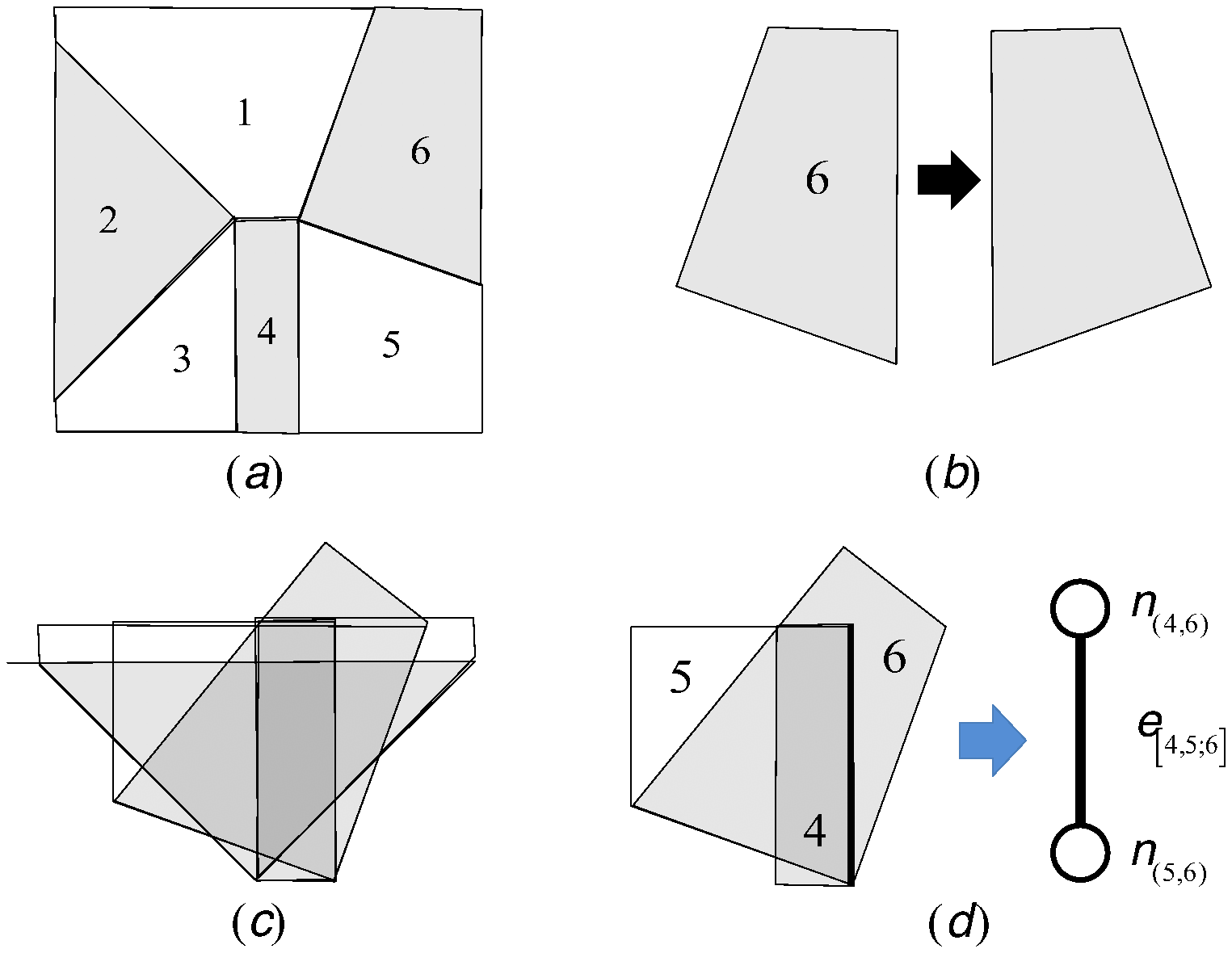}
\caption{(a) An example of a bicolored crease pattern known to be not flat-foldable, even with local flat-foldability around every interior vertices \cite{Hull_t}. 
  Each number in the figure represents the index of each polygon. (b) An example of the inversion of a shape of a polygon. The polygon whose index is $6$ in (a) is inverted. (c) The corresponding overlap diagram. The shapes of polygons with indices $1,3,5$ are the same as those in (a). The polygons with indices are $2,4,6$ are inverted compared to those in (a). (d)An edge corresponding to the non-self-intersection constraint which prohibits the polygon $6$ to intrude into the crease which connects the polygons $4$ and $5$.
}
\label{fig:ovDiag}
\end{figure}
Assuming a crease pattern satisfies the Kawasaki-Justin conditions at all interior vertices, ensuring no local angular breakdown, we introduce the overlap diagram.
This diagram is a planar arrangement of the constituent polygons that allows potential self-intersections (e.g., Fig. \ref{fig:ovDiag}(c), derived from the known non-foldable pattern in Fig. \ref{fig:ovDiag}(a) [3]).
It can be constructed in linear time with respect to the number of polygons.
By using this diagram, the flat-foldability problem is essentially reduced to finding a valid stacking order that resolves all potential intersections.
\begin{defin}\label{def:origamidiagram}
Overlap diagram : 
  An overlap diagram is a planar arrangement of the polygons formed by the creases in a given crease pattern,
  allowing any resulting potential self-intersections.
  The precise deterministic procedure for constructing this diagram is detailed in Algorithm \ref{alg:overlap_diagram_construction}.
  Note that the outer boundary of the crease pattern is not considered a crease.
\end{defin}
\begin{algorithm}
\caption{Construction of the Overlap Diagram}
We assume that the crease pattern is given, that is, all values of angles and edge-lengths of each polygon are given and the connection-relation among polygons are also given.
\label{alg:overlap_diagram_construction}
\begin{algorithmic}[1]
  \Require A crease pattern
  \Ensure An overlap diagram
  \State \textbf{Step 1:}Assign two colors to the constituent polygons of the crease pattern. (This 2-colorability is guaranteed by Lemma \ref{thm:bicolorability}).
  \State \textbf{Step 2:}Invert (reflect) the geometric shapes of the polygons belonging to one of the two color classes.
  \State \textbf{Step 3:}Arrange and connect the adjacent polygons along their shared creases on a plane, allowing (or ignoring) any resulting physical self-intersections. The resulting planar arrangement is defined as the overlap diagram.
\end{algorithmic}
\end{algorithm}
As a supplement to Definition \ref{def:origamidiagram}, we state as a lemma that any crease pattern that satisfies the necessary condition for locally flat folding at all interior vertices, as required by the Kawasaki-Justin theorem, is 2-colorable. See the appendix of this manuscript for the proof.

To avoid confusion with physical self-intersections, we strictly use the term "overlap" for polygons sharing planar regions in the overlap diagram.

A valid flat-folding requires a feasible vertical stacking order for these overlapping polygons, ensuring no polygon penetrates any crease. Such penetrations potentially occur when a polygon extends across a crease line. To systematically evaluate these constraints, we model the overlap diagram as an undirected graph, assigning a relative stacking order to each pair of overlapping polygons.

On each overlapping of the polygons in the overlap diagram, 
each relative order-in-stacking is assigned.
\begin{defin}\label{def:loc_ord}
  Let $P_i$ and $P_j$ be two polygons sharing a positive area in the overlap diagram.
  We define the relative stacking order $\sigma_{i,j} \in \{1, -1\}$, where $\sigma_{i,j} = 1$ if $P_i$ lies above $P_j$, and $-1$ if $P_i$ lies below $P_j$.
  By definition, the antisymmetry $\sigma_{i,j} = -\sigma_{j,i}$ holds.
\end{defin}
Physical constraints against self-intersections impose strict conditions on these variables:
\begin{lemma}\label{thm:constr}
  Let polygons $P_i$ and $P_j$ be connected by a crease, and let polygon $P_k$ overlap this crease in the overlap diagram. To prevent self-intersection, $P_i$ and $P_j$ must lie on the same side of $P_k$ in the vertical stacking order.
  In other words, $P_k$ cannot intrude between $P_i$ and $P_j$ in the stack.
  Consequently,
  \begin{eqnarray}
    \sigma_{i,k} \sigma_{k,j} = -1.\label{eq:two_rel}
  \end{eqnarray}
\end{lemma}

\begin{proof}
  If $P_k$ were to lie vertically between $P_i$ and $P_j$ (which implies $(\sigma_{i,k}, \sigma_{k,j}) = (1,1)$ or $(-1,-1)$, yielding $\sigma_{i,k}\sigma_{k,j} = 1$), it would penetrate the crease connecting $P_i$ and $P_j$, which is physically impossible.
\end{proof}
By Definition \ref{def:loc_ord}, Eq.(\ref{eq:two_rel}) can also be described by other three ways; $\sigma_{k,i}\sigma_{k,j}=1$, $\sigma_{i,k}\sigma_{j,k}=1$, and $\sigma_{k,i}\sigma_{j,k}=-1$.
This physical requirement constitutes the "non-intrusion condition" for the overlapping polygons.

As a consequence of Lemma \ref{thm:constr}, determining the relative stacking order at one overlap dictates the valid orders at adjacent overlaps.
This propagation of local constraints along chained connections naturally motivates the introduction of the following overlap graph.
  The non-intrusion constraint from Lemma \ref{thm:constr} links two relative stacking orders that share a common polygon.
  By treating these relative orders as nodes and the non-intrusion conditions as edges, we construct the overlap graph as followings.
  \begin{defin}\label{def:graph_rep}
    An overlap graph is an undirected graph that models the constraints of Lemma \ref{thm:constr} among polygons.
    Its formal structure, consisting of nodes and edges, is systematically constructed according to Algorithm \ref{alg:overlap_graph_construction}.
  \end{defin}
Note that the semicolon in $e_{[i,j;k]}$ distinguishes the intruding polygon $P_k$ from the crease-forming pair $(P_i, P_j)$, as their geometric roles are not interchangeable.
  \begin{algorithm}
    \caption{Construction of the Overlap Graph}
    \label{alg:overlap_graph_construction}
We assume that the crease pattern is given, that is, all values of angles and edge-lengths of each polygon are given and the connection-relation among polygons are also given.
\begin{algorithmic}[1]
  \Require An overlap diagram. (It is assumed that the crease pattern is given, including all angles, edge-lengths, and connection-relations among the constituent polygons).
\Ensure The corresponding overlap graph.

\Statex \textbf{Step 1: Node Generation}
\For{every pair of polygons $P_i$ and $P_j$ that share a positive overlapping area in the overlap diagram}
    \State Introduce a node $n_{(i,j)} \in V$.
    \State Define the relative stacking order between $P_i$ and $P_j$ as a variable $\sigma_{i,j} = -\sigma_{j,i}$ associated with this node.
\EndFor

\Statex \textbf{Step 2: Edge Generation}
\For{every polygon $P_k$ that geometrically spans across the crease connecting $P_i$ and $P_j$ in the overlap diagram}
    \State Introduce an undirected edge $e_{[i,j;k]} \in E$ connecting the nodes $n_{(i,k)}$ and $n_{(j,k)}$.
\EndFor
\end{algorithmic}
\end{algorithm}

This graph translates geometric stacking constraints into a combinatorial structure. For example, in Fig. \ref{fig:ovDiag}(d), edge $e_{[4,5;6]}$ encodes the constraint preventing polygon 6 from intruding into the crease between polygons 4 and 5.
The full overlap graph for the pattern in Fig. \ref{fig:ovDiag}(c) consists of 15 nodes and 20 edges as shown in Fig. \ref{fig:graph_hull_whole}(a).
In addition, as will be discussed later around Theorem \ref{thm:unfoldability} in Section \ref{sec:cycles_and_impossibility}, in the context of detecting impossibility of flat-folding, we focus on the cycles contained in the graph obtained with this convention.
By removing nodes and edges that are clearly unrelated to cycles from the graph in Fig. \ref{fig:graph_hull_whole}(a), the graph in Fig. \ref{fig:graph_hull_whole}(b) which contains 13 nodes and 18 edges is obtained.
\begin{figure}[h!tbp]\label{fig:graph_hull_whole}
  \centering
  \includegraphics[viewport=70 80 770 495,clip,width=2.4in]{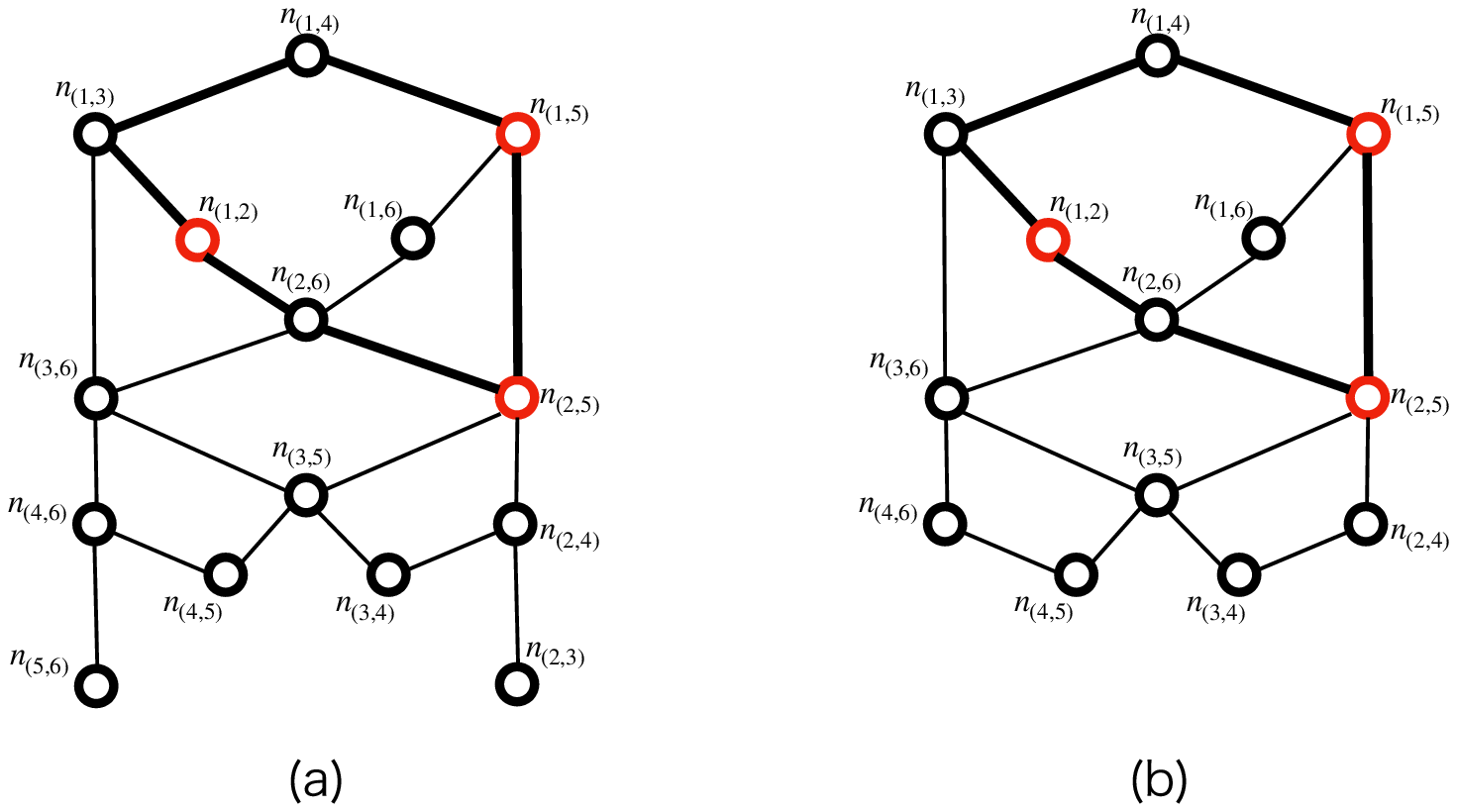}
  \caption{The overlap graph corresponding to the overlap diagram in Fig.1(b), which represents a cycle.
    A cycle corresponding to Fig.4 is described with thick line.
    The red circles are the second kind intermediations in the cycle.
    (a) Entire overap graph. (b) The overlap graph with nodes and edges, which are clearly unrelated from cycles, removed from the graph of (a).
  }
  \label{fig:graph_hull}
\end{figure}

\section{Cycles in Overlap Graph and Impossibility of Flat-Folding}\label{sec:cycles_and_impossibility}
In this section, we analyze cycles in the overlap graph to detect global non-foldability.
By tracing how local stacking constraints propagate along these cycles via Lemma \ref{thm:constr}, we can systematically identify logical contradictions in the stacking order.

To formalize this propagation, we define an intermediation as a subpath of length two within a cycle or path—specifically, a central node together with its two incident edge.
Because an intermediation is strictly defined by these incident edges, a single node can represent distinct intermediations depending on the cycle it belongs to.

Although purely combinatorial, intermediations encode the local geometric roles of polygons.
By tracing consecutive intermediations along a cycle, we may find that constraint propagation imposes two contradictory relative stacking orders on the same polygon pair, thereby proving non-foldability.
Based on the geometric relationships within this length-two path, intermediations are classified into two distinct types.
\begin{defin}\label{def:intermed}
  Let $n_{(i,j)}$ be a node in the overlap graph.
  When considered in the context of a specific path or cycle passing through it, the node $n_{(i,j)}$ itself is referred to as an \emph{intermediation} with respect to the specific pair of incident edges belonging to that path.
  Because its role depends on the chosen pair of incident edges, an intermediation is classified into two types based on those edges:
  \begin{itemize}
  \item \textbf{First kind:}  The two incident edges share the same spanning polygon \textup{(}e.g., $e_{[p,i;j]}$ and $e_{[i,q;j]}$\textup{)}.
  \item \textbf{Second kind:} The two incident edges have different spanning polygons \textup{(}e.g., $e_{[p,i;j]}$ and $e_{[j,q;i]}$\textup{)}.
  \end{itemize}
\end{defin}

Geometrically, this classification reflects whether the "non-intruding" spanning role remains fixed on $P_j$ (first kind) or alternates between $P_i$ and $P_j$ (second kind), as illustrated in Fig. \ref{fig:NFF_Hull}.
Because these roles dictate how the relative stacking order at $n_{(i,j)}$ restricts the adjacent order along a path, the constraint propagates deterministically.
This mechanism is formalized in Lemma \ref{thm:cor_local_order_taken_over}.
\begin{figure}[h!tbp]
  \centering
  \includegraphics[width=2.4in]{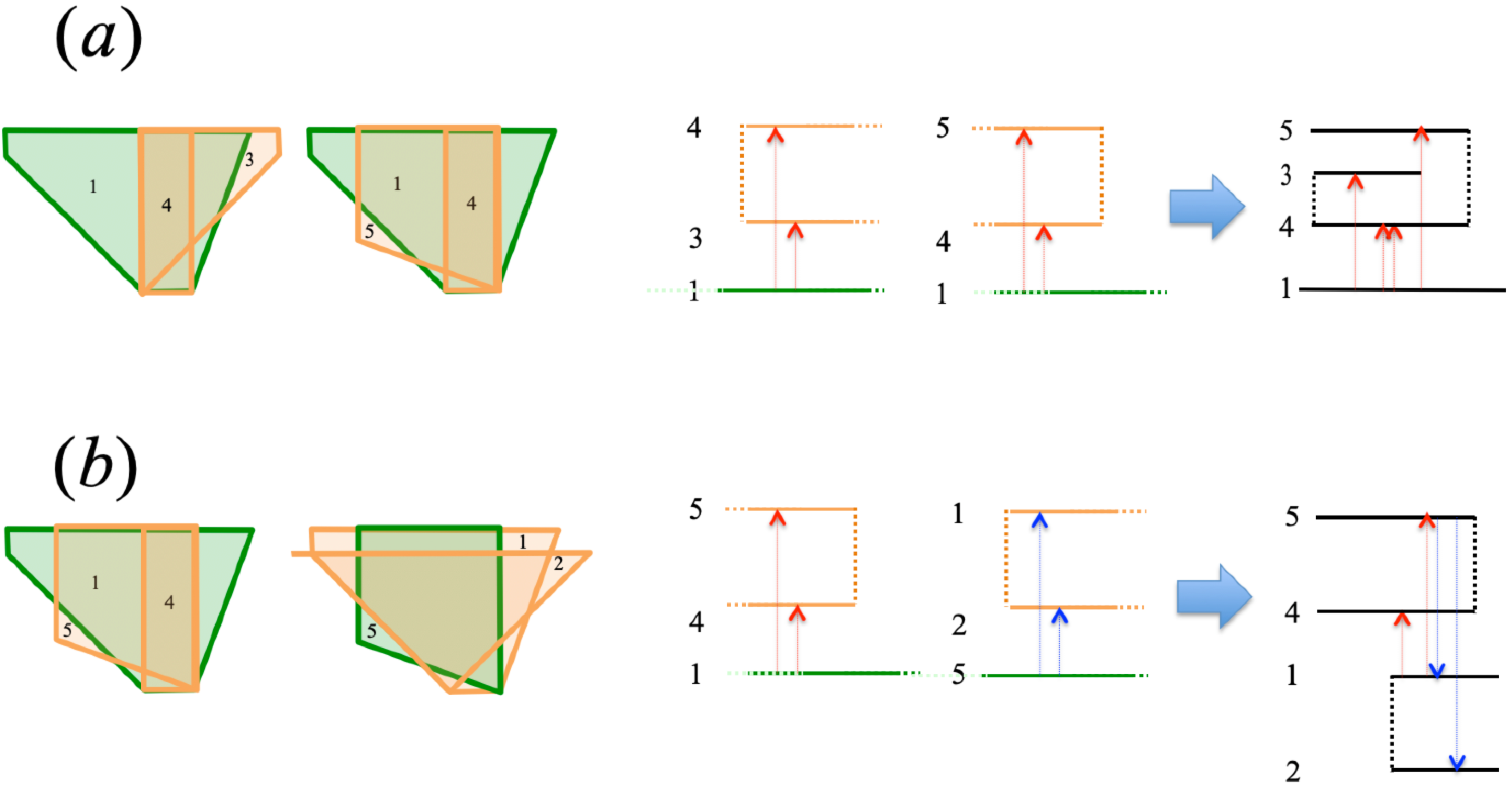}
  \caption{
    Explanation of the relationship between the kind of intermediation and the constraints imposed on the polygon stacking in the subgraph included in the overlap graph of Figure \ref{fig:graph_hull_whole}.
    (a) In the path $n_{(1,3)} \xrightarrow{e_{[3,4;1]}} n_{(1,4)} \xrightarrow{e_{[4,5;1]}} n_{(1,5)}$, $n_{(1,4)}$ is the intermediation of first kind.  (b) In the path $n_{(1,4)} \xrightarrow{e_{[4,5;1]}} n_{(1,5)} \xrightarrow{e_{[1,2;5]}} n_{(2,5)}$, $n_{(1,5)}$ is the intermediation of second kind. 
  }
  \label{fig:NFF_Hull}
\end{figure}
\begin{lemma}\label{thm:cor_local_order_taken_over}
  Consider a directed path traversing an intermediation at node $n_{(i,j)}$ from a preceding node $n_{(p,j)}$ to a succeeding node $n_{(q,t)}$, where $P_q$ is the newly appearing polygon and $t \in \{i, j\}$.
  By adopting the convention that the newly appearing polygon's index is placed first in the subscript (i.e., evaluating $\sigma_{q,t}$), the propagation of the stacking order is given by:
  \begin{itemize}
  \item If $n_{(i,j)}$ is of the first kind \textup{(}$t=j$\textup{)}, then $\sigma_{q,t} = \sigma_{i,j}$.
  \item If $n_{(i,j)}$ is of the second kind \textup{(}$t=i$\textup{)}, then $\sigma_{q,t} = -\sigma_{i,j}$.
  \end{itemize}
\end{lemma}
\begin{proof}
  Let the path arrive at $n_{(i,j)}$ via $e_{[p,i;j]}$.
  By Lemma \ref{thm:constr}, $\sigma_{p,j}\sigma_{j,i} = -1$, yielding $\sigma_{i,j} = \sigma_{p,j}$.
  If $n_{(i,j)}$ is of the first kind, the next edge is $e_{[i,q;j]}$.
  Lemma \ref{thm:constr} requires $\sigma_{i,j}\sigma_{j,q} = -1$.
  Using the antisymmetry $\sigma_{j,q} = -\sigma_{q,j}$, this simplifies to $\sigma_{q,j} = \sigma_{i,j}$.
  If $n_{(i,j)}$ is of the second kind, the next edge is $e_{[j,q;i]}$.
  Lemma \ref{thm:constr} requires $\sigma_{j,i}\sigma_{i,q} = -1$.
  Since $\sigma_{j,i} = -\sigma_{i,j}$ and $\sigma_{i,q} = -\sigma_{q,i}$, this implies $(-\sigma_{i,j})(-\sigma_{q,i}) = -1$, leading to $\sigma_{q,i} = -\sigma_{i,j}$.
\end{proof}

By repeatedly applying Lemma \ref{thm:cor_local_order_taken_over}, the stacking order propagates deterministically along any path in the overlap graph.
For a given crease pattern to be globally flat-foldable, completing the traversal of any cycle must return the propagation to the initial stacking order without contradiction.
Consequently, if traversing a cycle inverts the initial stacking order, it proves that the pattern is non-foldable, as it implies there is no valid stacking order.
This criterion is formalized in the following Theorem \ref{thm:unfoldability}.

At this point, it may be noticed that this propagation behavior is similar to the structure of a 2-XORSAT instance\cite{Schaefer}. 
Indeed, if we fix the subscript order of the variable $\sigma_{i,j}$ (which denotes the relative stacking order for each pair $(i,j)$), each pairwise constraint of Lemma \ref{thm:constr} symbolized by edges directly corresponds to a clause in 2-XORSAT.
Under this perspective, the cycle-parity criterion of Theorem \ref{thm:unfoldability} explained immediately below corresponds to the balance condition\cite{Harary} of 2-XORSAT.
Looking ahead, Theorem \ref{thm:property_of_junction}, which establishes the parity-under-addition statement, is equivalent to the GF(2)-linearity of the parity functional on the cycle space, and Theorem \ref{thm:cyc_poly}, which introduces the cycle-basis reduction to systematically pinpoint the polygons causing the non-foldability, also corresponds to a standard technique used in the context of 2-XORSAT\cite{Diestel_b}.
\begin{theorem}\label{thm:unfoldability}
  A consistent stacking order exists along a cycle $C$ in a given overlap graph if and only if $C$ contains an even number of second-kind intermediations.
\end{theorem}
\begin{proof}(Outline:)
  Let a cycle $C$ in the overlap graph be represented by a sequence $(n_0, e_1, n_1, \dots, e_L, n_L)$ with $n_L = n_0$. Here, the nodes and edges are denoted using representative notation.
    By Lemma \ref{thm:cor_local_order_taken_over}, traversing an intermediation at node $n_a$ propagates the relative stacking order $\sigma_a$ to the next node with a sign factor determined by the type of intermediation.
  Here, the indices of the nodes and edges are denoted in a representative manner.
  Although the details of the indices at each node and edge depend on the local connectivity (as captured by the notation $t$ in Lemma \ref{thm:cor_local_order_taken_over}), the sign propagation rule for $\sigma_a$ remains consistent by adopting the same subscript convention as in Lemma \ref{thm:cor_local_order_taken_over} at each step.
  Specifically, the sign inverts if and only if the intermediation is of the second kind.\\
  Let $\sigma_{a}$ be the relative stacking order at node $n_a$, indexed according to the convention in Lemma \ref{thm:cor_local_order_taken_over}) (placing the newly joined polygon's index first).
  Applying Lemma \ref{thm:cor_local_order_taken_over} successively along the cycle, the relationship between the initial order $\sigma^{(0)}$ and the order after $L$ steps $\sigma^{(L)}$ is given by:
  \begin{equation}
    \sigma_{L} = (-1)^{N_{\textrm{II}}} \sigma_{0},
  \end{equation}
  where $N_{\textrm{II}}$ denotes the total number of second-kind intermediations encountered along the cycle.
  Note that the kinds of the intermediations $n_{L-1}$ and $n_L (=n_0)$ are uniquely determined by the detail of the subscripts in edges $(e_{L-1}, e_L, e_1)$, ensuring that the kind of $n_{L-1}$ itself is accounted for in $N_{\textrm{II}}$.
  Since $n_L$ and $n_0$ represent the same physical constraint between the same pair of polygons, global consistency requires $\sigma_L = \sigma_0$. 
  This holds if and only if $N_{\textrm{II}}$ is even. 
  Conversely, an odd $N_{\textrm{II}}$ leads to $\sigma_0 = -\sigma_0$, a contradiction indicating that no valid stacking assignment exists.
\end{proof}
\begin{figure}[h!tbp]
  \centering
      \includegraphics[width=2.4in]{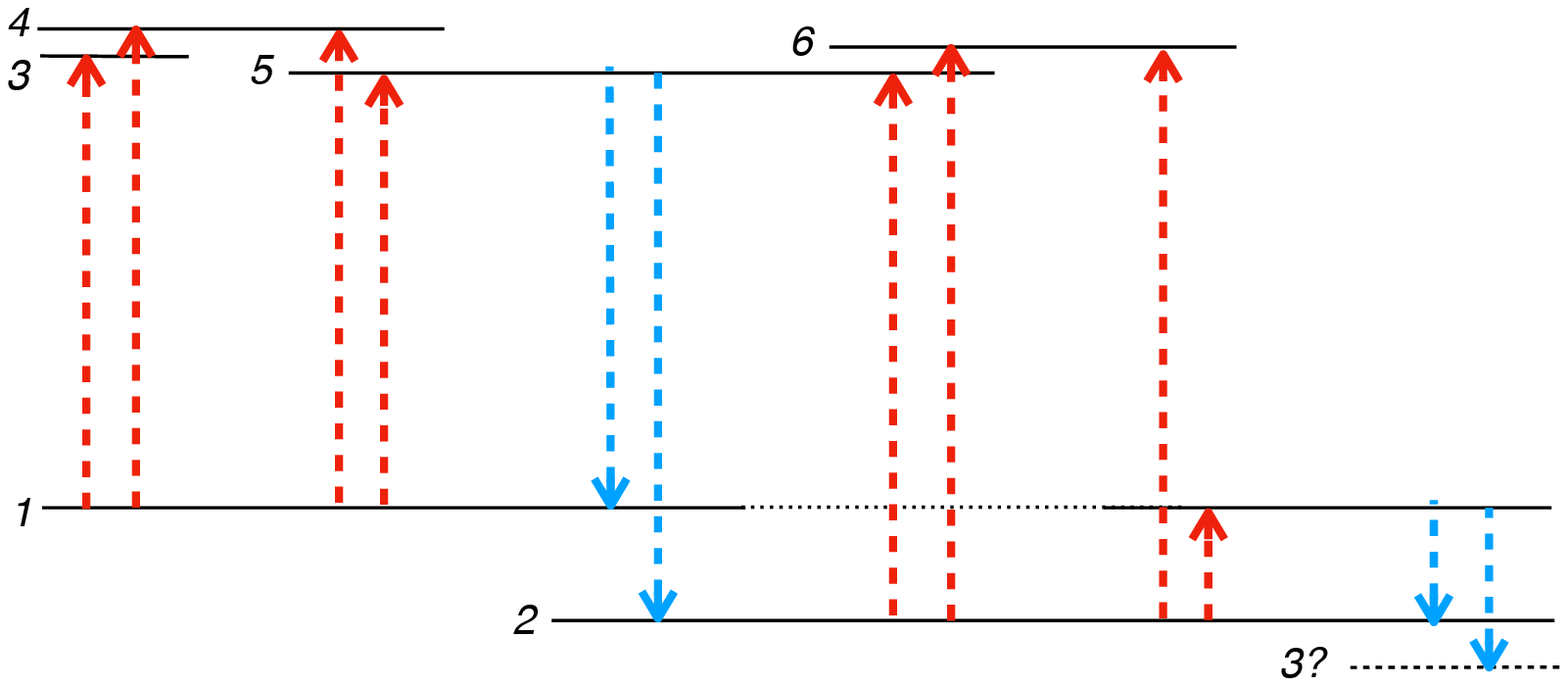}
  \caption{
A sequence of requiments of layer orders on the cycle of Eq.(\ref{eq:exam_cyc_v}) (and also of the thick line in Fig. \ref{fig:graph_hull}).
Initially a relative order $P_3>P_1$ is assumed. From that assumption, eventually $P_1>P_3$ is imposed.
  }
  \label{fig:mujyun_Hull}
\end{figure}

To demonstrate the detection of a self-inconsistent cycle, we examine the thick-lined cycle in Fig. \ref{fig:graph_hull} (schematized in Fig. \ref{fig:mujyun_Hull}).
Traversing this cycle clockwise from $n_{(1,3)}$, the sequence of nodes is given by:
\begin{equation}\label{eq:exam_cyc_v}
  n_{(1,3)} \to n_{(1,4)} \to n_{(1,5)} \to n_{(2,5)} \to n_{(2,6)} \to n_{(1,2)} \to n_{(1,3)}.
\end{equation}
To explicitly observe the condition propagation, let us assume the initial relative stacking order is $\sigma_{3,1} = +1$.
By successively applying the local constraint of Lemma 3 along the path, the sequence of relative stacking orders upon reaching $n_{(1,2)}$ is obtained as:
\begin{equation}
  \label{eq:loc_stck_exam}(\sigma_{4,1}, \sigma_{5,1}, \sigma_{2,5}, \sigma_{6,2}, \sigma_{1,2}) = (+1, +1, -1, +1, +1).
\end{equation}
The sign of the local stacking order inverts when the path passes through the second-kind intermediations $n_{(1,5)}$ and $n_{(2,5)}$.
At the final step, passing through the intermediation $n_{(1,2)}$ and edge $e_{[2,3;1]}$ completes the full cycle.
Because $n_{(1,2)}$ is also a second-kind intermediation, this step imposes $\sigma'_{3,1} = -1$. This contradicts the initial assumption $\sigma_{3,1} = +1$.
  This step-by-step tracking exemplifies the global criterion established in Theorem \ref{thm:unfoldability}.
  The cycle contains three ($N_{\textrm{II}} = 3$) second-kind intermediations.
  Because this number is odd, any initial assignment inevitably propagates to its own negation, proving that no valid stacking order exists for the crease pattern in Fig. 1(a). 
In statistical mechanics, such structural inconsistency within a cycle is known as frustration \cite{Nakajima_APPC, Nishimori, Nakajima_OSME}.

As is well known, any cycle in a given graph can be expressed as a linear combination ($\oplus$) of its cycle basis. 
As noted prior to Definition \ref{def:intermed}, the kind of intermediation at a node depends on the specific cycle traversing it.
Consequently, when combining two sub-cycles $C_1$ and $C_2$ to form a composite cycle $C_1 \oplus C_2$, the absolute number of second-kind intermediations is not simply conserved.
Since the intermediation attribute of each node can dynamically change depending on the combination of cycles, interpreting this behavior directly through the lens of 2-XORSAT, although ultimately correct, is not immediately straightforward.
Remarkably, however, the parity of this number is strictly conserved during cycle addition, which ultimately establishes the GF(2)-linearity of the parity functional on the cycle space\cite{Diestel_b}.
This crucial conservation of parity is established by the subsequent Theorem \ref{thm:property_of_junction}.
Its proof relies on two auxiliary results regarding the property changes of individual intermediations; to maintain the flow of our main argument, these supporting results are deferred to the Appendix as Lemma \ref{lem:property_change} and Corollary \ref{cor:parity_change_single}.
Because a cycle basis is computable in polynomial time, this parity conservation ensures that detecting a self-inconsistent cycle via Theorem \ref{thm:unfoldability} can be executed in polynomial time by verifying only the basis cycles.
To formalize this, let $C_1$ and $C_2$ be two cycles sharing a common path.
We partition the nodes in $C_1 \cup C_2$ into three disjoint subsets based on their incident cycle edges in $C_1 \oplus C_2$: $\mathcal{N}_{ex}$ (nodes excluded from $C_1 \oplus C_2$), $\mathcal{N}_{in}$ (nodes retaining their original incident edges), and $\mathcal{N}_{sw}$ (junction nodes where exactly one incident edge is replaced).
Note that $|\mathcal{N}_{sw}| = 2$ always holds.
\begin{theorem}\label{thm:property_of_junction}
  Let $C_1$ and $C_2$ be two cycles sharing a common path in a given overlap graph.
  The number of second-kind intermediations in the combined cycle $C_1 \oplus C_2$ satisfies:
  \begin{equation}
    N_{\textrm{II}}(C_1 \oplus C_2) \equiv N_{\textrm{II}}(C_1) + N_{\textrm{II}}(C_2) \pmod 2,
  \end{equation}
  where $N_{\textrm{II}}(C)$ denotes the total number of second-kind intermediations in a cycle $C$.
\end{theorem}
\begin{proof}
  (Outline :)
  We evaluate the parity of the difference $\Delta \equiv N_{\textrm{II}}(C_1 \oplus C_2) - N_{\textrm{II}}(C_1) \pmod 2$. This change in parity originates from two independent contributions:
  \begin{itemize}
  \item $\Delta_{set}$: The parity difference between the added second kind intermediations in $\mathcal{N}_{in}$ and the removed ones in $\mathcal{N}_{ex}$.
  \item$\Delta_{sw}$: The parity change due to the property transitions at the two junction nodes in $\mathcal{N}_{sw}$.
  \end{itemize}
  Thus, $\Delta \equiv \Delta_{set} + \Delta_{sw} \pmod 2$.
  By Corollary \ref{cor:parity_change_single} in Appendix, the parity change at the junctions $\Delta_{sw}$ is determined by whether the two nodes in $\mathcal{N}_{sw}$ have the same or different intermediation types in $C_2$.
  Meanwhile, $\Delta_{set}$ depends on the remaining nodes $C_2 \cap (\mathcal{N}_{in} \cup \mathcal{N}_{ex})$.
  Because the total number of second kind intermediations in $C_2$ is fixed, we have:
  \begin{equation}
    \Delta_{set} \equiv N_{\textrm{II}}(C_2) - N_{\textrm{II}}(C_2 \cap \mathcal{N}_{sw}) \pmod 2.
  \end{equation}
  By exhaustively analyzing the parity combinations based on $N_{\textrm{II}}(C_2)$ and the types of the two junction nodes in $\mathcal{N}_{sw}$, we obtain the relationships exhibited in Table \ref{tab:parity_cases}.
  \begin{table*}[t] 
    \centering
    \caption{Parity combinations for evaluating the difference $\Delta \equiv N_{\textrm{II}}(C_1 \oplus C_2) - N_{\textrm{II}}(C_1) \pmod 2$.}
    \label{tab:parity_cases}
    \begin{tabular}{llccc}
      \hline
      Parity of $N_{\textrm{II}}(C_2)$ & Types of two nodes in $\mathcal{N}_{sw}$ & $\Delta_{set}$ & $\Delta_{sw}$ & Total $\Delta \pmod 2$ \\
      \hline
      Odd  & Same (Both 1st or Both 2nd)  & Odd  & Even & Odd  \\
      Odd  & Different (One 1st, One 2nd) & Even & Odd  & Odd  \\
      Even & Same (Both 1st or Both 2nd)  & Even & Even & Even \\
      Even & Different (One 1st, One 2nd) & Odd  & Odd  & Even \\
      \hline
    \end{tabular}
  \end{table*}
  In all cases in Table \ref{tab:parity_cases}, the parity of the difference $\Delta$ strictly matches the parity of $N_{\textrm{II}}(C_2)$. Therefore, we conclude that $N_{\textrm{II}}(C_1 \oplus C_2) \equiv N_{\textrm{II}}(C_1) + N_{\textrm{II}}(C_2) \pmod 2$.
\end{proof}
The detailed proof of Theorem \ref{thm:property_of_junction} is shown in Appendix.

\begin{theorem}\label{thm:cyc_poly}
Detecting impossibility of flat-folding on a given graph representation of non-intrusion constraints can be performed in polynomial time.
\end{theorem}
\begin{proof}
  By Theorem \ref{thm:property_of_junction}, the parity of second-kind intermediations is conserved under cycle addition.
  Therefore, an overlap graph contains a self-inconsistent cycle (odd parity) if and only if at least one cycle in its basis possesses odd parity.
  This guarantees that the detection problem via Theorem \ref{thm:unfoldability} is reduced to evaluating solely the cycles within a cycle basis.

  A cycle basis for a given graph can be found in $O(EV)$ time \cite{KLMMRUZ}, where $V$ and $E$ are the numbers of nodes and edges, respectively.
  Let $N$ be the total number of polygons in the crease pattern.
  Since each node represents a pairwise constraint between two polygons, the upper bound on the number of nodes is $V = O(N^2)$.
  Consequently, the number of edges is bounded by $E = O(V^2) = O(N^4)$.

  Thus, the overall computational time to extract the cycle basis and verify the self-inconsistency condition is bounded by $O(EV) = O(N^6)$. 
\end{proof}

\section{Detection of Impossibility by Cycles in Overlap Graph: Applicability and Limitations}
While our cycle-based criterion efficiently identifies non-foldability, the absence of an inconsistent cycle does not guarantee global flat-foldability, reflecting the fact that testing flat-foldability is generally NP-hard \cite{BH, Akitayaetal}.
Our current framework models pairwise non-intrusion constraints but omits two specific geometric constraints.
First, collinear parallel creases sharing a line segment in an overlap diagram impose simultaneous stacking constraints on four adjacent polygons \cite{Nakajima_APPC,Nakajima_OSME}.
This four-body interaction is fundamentally different from our pairwise non-intrusion condition (Definition \ref{def:graph_rep}).
Second, when three or more polygons simultaneously share a single region in an overlap diagram, preventing a cyclic stacking order requires a different type of pairwise constraint \cite{Nakajima_APPC, Nakajima_OSME} that is not captured by our unintrudability condition.

Nevertheless, our pairwise non-intrusion constraint represents a fundamental physical requirement.
Because any local inconsistency detected by our condition cannot be resolved by incorporating those unmodeled constraints, a single self-inconsistent cycle found in the current overlap graph strictly prohibits global flat-foldability.
Therefore, detecting such a self-inconsistent cycle provides a robust sufficient condition for the impossibility of flat-folding.
The broader implications of this result are twofold.
First, the overlap graph provides a universal criterion for efficiently detecting non-foldability that does not rely on any specific geometric regularities or symmetric structures in the given crease pattern.
Second, from the perspective of computational complexity, this approach successfully identifies a polynomially tractable subclass of non-foldable instances within the broader NP-hard problem of global flat-foldability, even when local flat-foldability is strictly guaranteed.

\section{Summary and Conclusions}
In this manuscript the combination of unintrudabilities of a polygon into a crease consisting of two polygons is focused on and a sufficient condition for the impossibility of flat-folding is proven.
First each unintrudability condition is redescribed as the graphical representation which is in this paper called overlap graph. 
Next on the graph it is shown that the relationship among the relative orders-in-stacking between the polygons involved in unintrudability constraints are imposed. 
Then it is proven that there are graphs with a cycle such that no combination of relative orders-in-stacking can satisfy in total the required relationship on it.

The algorithmic ability to detect the impossibility of flat-folding also provides one step deeper scope in discussing the computational complexity of the flat-foldability determination. 
In general the computational complexity class is determined by the worst-case complexity of instances in the problem. 
Then, if the instances which are impossible to satisfy the given conditions are detected in polynomial time, they cannot be the worst-case ones, in the sense that they are easy to determine to be impossible. 
From the results of this paper, we can conclude that the worst-case hardness of the flat-foldability problem does not come from the combinations of unintrudability conditions, but from avoiding the cyclic order of stacking of three polygons or choosing stacking for inclusionary fold pairs. 
Thus, the studies and the methods in this paper are considered to be useful for separating the easy crease patterns from those with the worst-case computational complexity, and for identifying the properties where the hardness comes from.

\small
\bibliographystyle{abbrv}
\bibliography{samplerefs}

\section*{Appendix}
\subsection*{A lemma for defining overlap graph}
\begin{lemma}\label{thm:bicolorability}
The arrangement of polygons in the crease patterns that satisfy for all interior vertices the necessary conditions imposed by the Kawasaki-Justin theorem is bicolorable.
\end{lemma}
\begin{proof}
First, one of the sufficient conditions for a given graph to be bipartite is that it has no odd cycles.
Required by Kawasaki-Justin theorem, for any interior vertex in a locally flat foldable crease pattern the number of angles that converge around it shall be even.
When every crease in a given crease pattern is regarded as an edge in a graph and the crease pattern itself is considered as an entire graph, the above requirement from Kawasaki-Justin theorem necessary to be locally flat-foldable is equivalent to that every cycle around each vertex on the plane-dual graph of the crease pattern shall be even.
In a planar graph, the set of cycles that surround every vertex is a cycle base.  Therefore it is resulted that every cycle in the graph is even.
Then, the planar dual of the given crease pattern is bipartite.  It means that the crease pattern is bicolorable.
\end{proof}
Please note that the correspondence between the crease patterns and graphs used in this proof is different from the overlap graphs in the main content of the paper.

\subsection*{Supplementary note on the introduction of relative ordering for overlaps}\label{sec:cyclic_order}
On the Introduction of Relative Ordering in Overlap Diagrams : \\
In considering the overlap of polygons within an overlap diagram, one might naturally attempt to define an ordering by assigning a global overlap level $l_i$ to each polygon $i$, and deriving local relations from the difference $l_i-l_j$.
However, this approach is inadequate in general, since overlap diagrams may contain cycles in their stacking relations, thereby precluding a consistent global ordering.
To accommodate such cyclic structures, we instead define a relative ordering as a relation that captures the relative overlap level (hierarchical position) between overlapping polygons, but only at the locations where the overlaps actually occur. This formulation allows us to describe stacking configurations that cannot be represented by global levels alone.\\
\begin{figure}[h!tbp]\label{fig:AB}
  \includegraphics[width=2.4in]{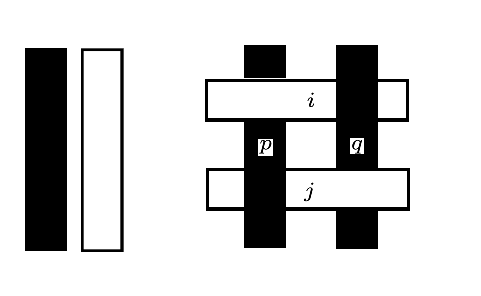}
  \caption{Scematic diagram of an example of a cyclic stacking. Four rectangular polygons of the same shape $i,j,p,q$ overlap in a cyclic order. This situation can be expressed by the variables of relative order-of-stacking as $(\sigma_{p,i},\sigma_{p,j},\sigma_{q,i},\sigma_{q,j})=(1,-1,-1,1)$, however it is not possible to define a global overlap order among this four rectangules.}
\end{figure}
In addition, when a single pair of polygons overlaps directly at two or more distinct points on the overlap diagram, the relative ordering framework introduced in this paper becomes ill-defined. 
To circumvent this issue, we restrict our attention to crease patterns in which all polygonal regions are convex. 
Consequently, the method presented here does not apply to diagrams that include non-convex (i.e., concave) polygons. At present, it is unknown whether such crease patterns exist—particularly those that are flat-foldable, as far as the author's knowledge.

\subsection*{Detailed proof for Theorem \ref{thm:unfoldability}}
\begin{proof}(Detailed:)
  We trace the relationship between the number of second-kind intermediations contained in the cycle and the constraints propagated around the cycle by dividing the cases into those with the path fragments contained in a cycle being the following four cases;
  \begin{itemize}
  \item $n_{(y,z)}$,$e_{[y,a;z]}$,$n_{(z,a)}$, $e_{[z,b;a]}$, $n_{(a,b)}$, $e_{[a,c;b]}$, $n_{(b,c)}$ 
  \item $n_{(y,z)}$,$e_{[y,b;z]}$,$n_{(z,b)}$, $e_{[z,a;b]}$, $n_{(a,b)}$, $e_{[a,c;b]}$, $n_{(b,c)}$ 
  \item $n_{(y,a)}$,$e_{[y,z;a]}$,$n_{(z,a)}$, $e_{[z,b;a]}$, $n_{(a,b)}$, $e_{[a,c;b]}$, $n_{(b,c)}$ 
  \item $n_{(y,b)}$,$e_{[y,z;b]}$,$n_{(z,b)}$, $e_{[z,a;b]}$, $n_{(a,b)}$, $e_{[a,c;b]}$, $n_{(b,c)}$ 
  \end{itemize}
  The indices of the above four path fragments are described in unified manner as follows;
  \begin{eqnarray}
    e_{[y,u;\bar{u}]}, n_{(z,t)}, e_{[z,\bar{t};t]}, n_{(a,b)}, e_{[a,c;b]},
  \end{eqnarray}
  where $t \in \mathcal{T}=\{a,b\}$ and $u \in \mathcal{U}=\{z,t\}$.
  Note that the set $\mathcal{U}$ is defined depending on the choice of $t$, namely, the contents of $\mathcal{T}$.
  In addition, we define $\bar{t}\in\mathcal{T}/t$ and $\bar{u}\in\mathcal{U}/u$.\\  
  Since the cycle is closed, $n_{(\bar{u},u)}$ and $n_{(b,c)}$ are connected by another path.
  Suppose that the values $k$ and $N'_{\textrm{II}}$ are respectively the indication of the second kind intermediation of the node $n_{(\bar{t},t)}=n_{(a,b)}$ and the count of the second-kind intermediations along the another path between $n_{\bar{t},t}$ and $n_{(u,\bar{u})}=n_{(z,t)}$, excluding both endnodes themselves.
  By repeatedly applying Lemma \ref{thm:cor_local_order_taken_over} to this another path between $n_{(\bar{t},t)}$ and $n_{(u,\bar{u})}$, we obtain
  \begin{eqnarray}                               
    \sigma_{u,\bar{u}}=(-1)^{k+N'_\textrm{II}}\sigma_{\bar{t},t}, \label{eq:middle_of_cycle_Re}
  \end{eqnarray}
  where $t$ in the subscript of $\sigma_{\bar{t},t}$ is given by the choice of the edge $e_{[z,\bar{t};t]} \in \{e_{[z,a;b]},e_{[z,b;a]}\}$, which is a final one on the traversing.
  Here $N'_\textrm{II}$ denotes the second-kind intermediations encountered along the path excluding both endpoints.
  Note that the denoting manner of the indices of the relative stacking-orders, $i$ and $j$ in $\sigma_{i,j}$, in this proof are arranged with the same manner as noted in the last part of Lemma \ref{thm:cor_local_order_taken_over}.
  At the final step of traversing, when we return to the starting node $n_{(a,b)}$ via the edge $e_{[z,\bar{t};t]}$, Lemma \ref{thm:constr} yields
  \begin{eqnarray}\label{eq:last_edge_Re}
    \sigma^{\prime}_{t,\bar{t}}=\sigma_{z,\bar{t}}.
  \end{eqnarray}
  Substituting Eq.~(\ref{eq:middle_of_cycle_Re}) into Eq.~(\ref{eq:last_edge_Re}) gives
  \begin{eqnarray}\label{eq:full_cycle_CasesRep_Re}
    \sigma^{\prime}_{t,\bar{t}}=
    \begin{cases}
      (-1)^{k+N'_\textrm{II}}\sigma_{\bar{t},t} \quad \text{($u=z$)},\\
      -(-1)^{k+N'_\textrm{II}}\sigma_{\bar{t},t} \quad \text{($u=t$)}.
    \end{cases}
  \end{eqnarray}
  Introducing $k^{\prime}$ as the indicator of the second-kind intermediation at $n_{(u,\bar{u})}=n_{(z,t)}$, this can be rewritten as
  \begin{eqnarray}\label{eq:full_cycle_SingleRep_Re}
    \sigma^{\prime}_{\bar{t},t}=(-1)^{k+N'_\textrm{II}+k^{\prime}}\sigma_{\bar{t},t}.
  \end{eqnarray}
  Here $k$ and $k'$ are determined by the choice of $(t,u)$ as follows:
  \begin{eqnarray}
    (k,k^{\prime})=
    \begin{cases}
      (1,1) \quad \text{($(t,u)=(a,t)$)}\\ 
      (0,1) \quad \text{($(t,u)=(b,t)$)}\\ 
      (1,0) \quad \text{($(t,u)=(a,z)$)}\\ 
      (0,0) \quad \text{($(t,u)=(b,z)$)}. 
    \end{cases}
  \end{eqnarray}
  Thus the exponent $k+N'_\textrm{II}+k'$ counts precisely all second-kind intermediations in the cycle.
  Therefore the relationship between the types of intermediates and the propagation of conditions described in Corollary also holds for closed paths.
  Consistency requires $\sigma'_{\bar{t},t}=\sigma_{\bar{t},t}$, which holds if and only if this number is even.
  Therefore, a consistent stacking assignment exists exactly when every cycle in the overlap graph contains an even number of second-kind intermediations, as claimed.
\end{proof}

\subsection*{Detailed proof for Theorem \ref{thm:property_of_junction}}
\begin{proof}
  We evaluate the difference $N_{\textrm{II}}(C_1 \oplus C_2)-N_{\textrm{II}}(C_1)$ by considering cases based on the parity of $N_{\textrm{II}}(C_2)$.
  The change in parity results from the change in the intermediation property of the two nodes in $\mathcal{N}_{sw}$, as well as the difference in the number of intermediations in $C_2 \cap \mathcal{N}_{in}$ and $C_2 \cap \mathcal{N}_{ex}$.

  First, consider the case where $N_{\textrm{II}}(C_2)$ is odd and the two intermediations in $\mathcal{N}_{sw}$ are either both first-kind or both second kind in $C_2$.
  In this case, the set $C_2 \cap (\mathcal{N}_{in} \cup \mathcal{N}_{ex})$ contains an odd number of second kind intermediations.
  These are distributed between $C_2 \cap \mathcal{N}_{in}$ and $C_2 \cap \mathcal{N}_{ex}$ such that one subset contains an odd number and the other contains an even number.
  Consequently, the difference between $N_{\textrm{II}}(C_2 \cap \mathcal{N}_{in})$ and $N_{\textrm{II}}(C_2 \cap \mathcal{N}_{ex})$ has an odd parity.
  In addition, the contribution from the change of properties in $\mathcal{N}_{sw}$ is determined to be even by Corollary \ref{cor:parity_change_single}.
  Alternatively, consider the case where $N_{\textrm{II}}(C_2)$ is odd, but of the two intermediations in $\mathcal{N}_{sw}$, one is of the first-kind and the other is of the second kind.
  In this scenario, the set $C_2 \cap (\mathcal{N}_{in} \cup \mathcal{N}_{ex})$ contains an even number of second kind intermediations.
  These are distributed between $C_2 \cap \mathcal{N}_{in}$ and $C_2 \cap \mathcal{N}_{ex}$ such that both subsets contain either an odd number or an even number.
  Therefore, the difference between $N_{\textrm{II}}(C_2 \cap \mathcal{N}_{in})$ and $N_{\textrm{II}}(C_2 \cap \mathcal{N}_{ex})$ has an even parity.
  In addition, the contribution from the change of properties in $\mathcal{N}_{sw}$ is determined to be odd by Corollary \ref{cor:parity_change_single}.
  Hence, we conclude that the difference $N_{\textrm{II}}(C_1 \oplus C_2)-N_{\textrm{II}}(C_1)$ yields an odd parity in all cases where $N_{\textrm{II}}(C_2)$ is odd.

  Next, we consider the case where $N_{\textrm{II}}(C_2)$ is even. Suppose first that the two intermediations in $\mathcal{N}_{sw}$ share the same property in $C_2$ (i.e., both are first-kind or both are second kind). In this case, the set $C_2 \cap (\mathcal{N}_{in} \cup \mathcal{N}_{ex})$ contains an even number of second kind intermediations. These are distributed such that both $C_2 \cap \mathcal{N}_{in}$ and $C_2 \cap \mathcal{N}_{ex}$ contain either an even or an odd number, yielding an even parity for their difference. In addition, the contribution from the change of properties in $\mathcal{N}_{sw}$ is determined to be even by Corollary \ref{cor:parity_change_single}.

Conversely, if the two intermediations in $\mathcal{N}_{sw}$ differ in property (one is first-kind and the other is second kind), the set $C_2 \cap (\mathcal{N}_{in} \cup \mathcal{N}_{ex})$ contains an odd number of second kind intermediations. Consequently, one subset must contain an odd number and the other an even number, which gives their difference an odd parity. Since the contribution from the change of properties in $\mathcal{N}_{sw}$ is also determined to be odd by Corollary \ref{cor:parity_change_single}, the sum of these two odd parities is even.

Hence, we conclude that the difference $N_{\textrm{II}}(C_1 \oplus C_2)-N_{\textrm{II}}(C_1)$ yields an even parity in all cases where $N_{\textrm{II}}(C_2)$ is even.

  Thus we conclude that the parity of the difference $N_{\textrm{II}}(C_1 \oplus C_2)-N_{\textrm{II}}(C_1)$ always matches the parity of $N_{\textrm{II}}(C_2)$. 
\end{proof}

\subsection*{Lemmas for Theorem \ref{thm:property_of_junction}}\label{sec:sppl_lemmas}
\begin{lemma}\label{lem:property_change}
  Let $C_1$ and $C_2$ be two cycles, and let $C_1 \oplus C_2$ denote the joined cycle. For any node $n_{(x,y)}$ shared by both cycles, if its intermediation property differs between $C_1$ and $C_2$, it acts as a second kind intermediation in $C_1 \oplus C_2$. If the property is the same in both cycles, it acts as a first-kind intermediation in $C_1 \oplus C_2$.
\end{lemma}
\begin{proof}
  Suppose $n_{(x,y)} \in \mathcal{N}_{sw}$ acts as a second kind intermediation in $C_2$. The transition from $C_2$ to $C_1 \oplus C_2$ alters the path through $n_{(x,y)}$ according to one of the following two transformations:
  one is from $n_{(v,x)} \xrightarrow{e_{[v,x;y]}} n_{(x,y)} \xrightarrow{e_{[x,z;y]}} n_{(y,z)}$ to $n_{(v,x)} \xrightarrow{e_{[v,x;y]}} n_{(x,y)} \xrightarrow{e_{[w,y;x]}} n_{(w,x)}$,
  and another is from $n_{(v,y)} \xrightarrow{e_{[v,x;y]}} n_{(x,y)} \xrightarrow{e_{[x,z;y]}} n_{(y,z)}$ to
  $n_{(v,x)} \xrightarrow{e_{[v,x;y]}} n_{(x,y)} \xrightarrow{e_{[w,y;x]}} n_{(w,x)}$,
  where the polygon index $v$ can represent any polygon in the unfolding other than $x$ and $y$.
  In any configuration of the first case (i.e., $v=w$, $v=z$, or $v \notin \{w,z\}$), the polygon index shared by $n_{(v,x)}$ and $n_{(x,y)}$ is invariably $x$. Thus, this structural change shifts $n_{(x,y)}$ from a second kind to a first-kind intermediation, inducing a parity change of $1$.
  Analogous reasoning applies to the second path transformation, which also yields a parity change of $1$.
  
Therefore, when $n_{(x,y)}$ is a second kind intermediation in $C_2$, its property change contributes exactly $1$ to the value of $N_{\textrm{II}}(C_1 \oplus C_2) - N_{\textrm{II}}(C_2)$.
\end{proof}

Lemma \ref{lem:property_change} can be restated as follows,
\begin{cor}\label{cor:parity_change_single}
  Consider the joined cycle $C_1 \oplus C_2$ formed from cycles $C_1$ and $C_2$. For any shared node $n_{(x,y)}$, if it is first-kind intermediation in $C_2$, its intermediation-property in $C_1$ is conserved in $C_1 \oplus C_2$.
  If it is second kind intermediation in $C_2$, its property is altered, making it different in $C_1 \oplus C_2$ from that in $C_1$.
\end{cor}
\begin{proof}
  On the case that $n_{(x,y)}$ is first-kind in $C_2$, when $n_{(x,y)}$ is first-kind in $C_1$ it is also first-kind in $C_1 \oplus C_2$ and when it is second kind in $C_1$ then it is also second kind in $C_1 \oplus C_2$.
   While, on the case that $n_{(x,y)}$ is second kind in $C_2$, when $n_{(x,y)}$ is first-kind in $C_1$ it is also second kind in $C_1 \oplus C_2$ and when it is second kind in $C_1$ then it is also first-kind in $C_1 \oplus C_2$.
\end{proof}

\end{document}